\documentclass[11pt]{article}

\usepackage{amsmath,amsfonts,amssymb,amsthm,epsfig,epstopdf,titling,url,array}
\usepackage[left=1in, right=1in, top=1in, bottom=1in]{geometry}
\usepackage{cite}
\usepackage{hyperref}
\usepackage{setspace}
\usepackage{graphicx}
\usepackage{float}
\usepackage{tikz-cd}
\usepackage{tikz}
\usepackage{color}
\usetikzlibrary{decorations.pathreplacing}
\usetikzlibrary{quotes}

\theoremstyle{plain}
\newtheorem{Theorem}{Theorem}

\newtheorem{Proposition}[Theorem]{Proposition}
\newtheorem{Corollary}{Corollary}

\newtheorem{Definition}{Definition}
\newtheorem{Remark}{Remark}

\newcommand{\R}{\mathbb{R}}
\renewcommand{\d}{\textnormal{ d}}

\newcommand{\hidethis}[1]{}

\newcommand{\E}{\mathbb{E}}
\newcommand{\support}{\textnormal{supp}}
\newcommand{\ren}{\textnormal{Ren}}
\newcommand{\prob}{\mathbb{P}}
\newcommand{\essup}{\textnormal{ess sup}}

\begin{document}
\title{Conditional R\'enyi entropy and the relationships between R\'enyi capacities}
\author{
  Gautam Aishwarya\\
  \texttt{University of Delaware}\\
  \texttt{gautama@udel.edu}
  \and
 Mokshay Madiman\\
 \texttt{University of Delaware}\\
  \texttt{madiman@udel.edu}
}

\date{}
\maketitle

\abstract{The analogues of Arimoto's definition of conditional R\'enyi entropy and R\'enyi mutual 
information are explored for abstract alphabets. 
These quantities, although dependent on the reference measure, have some useful properties 
similar to those known in the discrete setting. In addition to laying out some such basic properties
and the relations to R\'enyi divergences, the relationships between the families of mutual informations 
defined by Sibson, Augustin-Csisz\'ar, and Lapidoth-Pfister, as well as the corresponding capacities, are explored.}

\section{Introduction}

Shannon's information measures (entropy, conditional entropy, the Kullback divergence or relative entropy, and mutual information) 
are ubiquitous both because they arise as operational fundamental limits of various communication or statistical inference problems,
and because they are functionals that have become fundamental in the development and advancement of probability theory itself.
Over a half century ago, R\'enyi \cite{Ren61} introduced a family of information measures extending those of Shannon \cite{Sha48},
parametrized by an order parameter $\alpha\in [0,\infty]$.
R\'enyi's information measures are also fundamental-- indeed, they are (for $\alpha>1$) just  monotone functions of $L^p$-norms,
whose relevance or importance in any field that relies on analysis need not be justified. Furthermore,
they show up in probability theory, PDE, functional analysis, additive combinatorics, and convex geometry 
(see, e.g., \cite{Tos14:2, WM14, MWW19, WWM14:isit, MWW17:1, MMX17:0, FLM20}), in ways where understanding
them as information measures instead of simply as monotone functions of $L^p$-norms is fruitful. 
For example, there is an intricate story of parallels between entropy power inequalities (see, e.g., \cite{Sta59, MB06:isit, MG19}),
Brunn-Minkowski-type volume inequalities (see, e.g., \cite{BMW11, FMMZ18, MK18}) and sumset cardinality (see, e.g., \cite{Ruz09:1, MMT10:itw, Tao10, MMT12, KM14}),
which is clarified by considering logarithms of volumes and Shannon entropies as members of the larger class of R\'enyi entropies.
It is also recognized now that R\'enyi's information measures show up as fundamental operational limits in a range of
information-theoretic or statistical problems (see, e.g., \cite{Csi95, Ari96, BL14, BL14:itw, TH18:1, TH18:2, LP19}).
Therefore, there has been considerable interest in developing the theory surrounding R\'enyi's information measures
(which is far less well developed than the Shannon case), and there has been a steady stream 
of recent papers \cite{TMA12, EH14, FB14, Ver15:ita, LP18:itw, LP19, SV18:1, SV18:2, SM19, Nak19:1} elucidating their properties 
beyond the early work of \cite{Cam65, Sib69, Ari77, Aug78:hab}. This paper, part of which was presented
at ISIT 2019 \cite{AM19:isit}, is a further contribution along these lines.

%ALSO: \cite{KKP08:isit, PV10:allerton, CV14:isit, SI17} and \cite{Csi67, Csi72} 

More specifically, three notions of R\'enyi mutual information have been considered in the literature 
(usually named after Sibson, Arimoto and Csisz\'ar) for discrete alphabets. 
Sibson's definition has also been considered for abstract alphabets, but Arimoto's definition has not.
Indeed Verd\'u  \cite{Ver15:ita} asserts: ``One shortcoming of Arimoto's proposal is that its generalization to non-discrete alphabets is not self-evident.''
The reason it is not self-evident is because although there is an obvious generalized definition, the mutual information arising from this 
notion depends on the choice of reference measure on the abstract alphabet, which is not a desirable property. Nonetheless, the perspective taken in this note is that
it is still interesting to develop the properties of the abstract Arimoto conditional R\'enyi entropy keeping in mind
the dependence on the reference measure. The Sibson definition is then just a special 
case of the Arimoto definition where we choose a particular, special reference measure.  

Our main motivation comes from considering various notions of R\'enyi capacity.
While certain equivalences have been shown between various such notions by Csisz\'ar \cite{Csi95} for finite alphabets and Nakibo\u{g}lu \cite{Nak19:1, Nak19:2} for abstract alphabets,
the equivalences and relationships are further extended in this note. 

This paper is organized in the following manner. In Section \ref{sec:defn} below we begin by defining conditional R\'enyi entropy for random variables taking values in a Polish space. 
Section \ref{sec:div} presents a variational formula for the conditional R\'enyi entropy in terms of R\'enyi divergence, which will be a key ingredient in several results later. 
Basic properties that the abstract conditional R\'enyi entropy satisfies akin to its discrete version are proved in Section \ref{sec:pty}, including description for special orders $0,1$ and $\infty$, 
monotonicity in the order, reduction of entropy upon conditioning, and a version of the chain rule. 
Section \ref{sec:mi} discusses and compares several notions of $\alpha$-mutual information. 
The various notions of channel capacity arising out of different notions of $\alpha$-mutual information 
are studied in Section \ref{sec:cc}, which are then compared using results from the preceding section.

%nakpfisedt%
\section{Definition of conditional R\'enyi entropies} \label{sec:defn}

Let $S$ be a Polish space and $\mathcal{B}_{S}$ its Borel $\sigma$-algebra. We fix a $\sigma$-finite reference measure $\gamma$ on $(S, \mathcal{B}_{S})$.
Our study of entropy and in particular all $L^p$ spaces we talk about will be with respect to this measure space, unless stated otherwise.

\begin{Definition}
Let $X$ be an $S$-valued random variable with density $f$ with respect to $\gamma$. We define the R\'enyi entropy of $X$ of order $\alpha \in (0,1) \cup (1, \infty)$ by
\[
h^{\gamma}_{\alpha}(X) = \frac{\alpha}{1- \alpha } \log \Vert f \Vert_{\alpha}.
\]
\end{Definition} 
It will be convenient to write down R\'enyi entropy as
\[
h^{\gamma}_{\alpha}(X) = - \log \ren^{\gamma}_{\alpha}(X),
\]
where $\ren^{\gamma}_{\alpha}(X) = \Vert f \Vert_{\alpha}^{\frac{\alpha}{\alpha - 1}}$ will be called the R\'enyi probability of order $\alpha$ of $X$.

Let $T$ be another Polish space with a fixed measure $\eta$ on its Borel $\sigma$-algebra $\mathcal{B}_{T}$. Now suppose $X,Y$ are, resepectively, $S,T$-valued random variables with a joint density $F: S \times T \to \mathbb{R}$ w.r.t. the reference $\gamma \otimes \eta$. We will denote the marginals of $F$ on $S$ and $T$ by $f$ and $g$ respectively. This in particular means that $X$ has density $f$ w.r.t. $\gamma$ and $Y$ has density $g$ w.r.t. $\eta$. 
Just as R\'enyi probability of $X$, one can define the R\'enyi probability of the conditional $X$ given $y=y$ by the expression $\ren^{\gamma}_{\alpha}(X \vert Y=y) = \Vert \frac{F(\cdot, y)}{g(y)} \Vert^{\frac{\alpha}{\alpha - 1}}_{\alpha}$. The following generalizes \cite[Definition 2]{FB14}. 
\begin{Definition}
Let $\alpha \in (0,1) \cup (1, \infty)$. We define the conditional R\'enyi entropy $h^{\gamma}_{\alpha}(X \vert Y)$ in terms of a weighted mean of conditional R\'enyi probabilities $\ren^{\gamma}_{\alpha}(X \vert Y = y)$,
\[
h^{\gamma}_{\alpha}(X \vert Y) = - \log \ren^{\gamma}_{\alpha} (X \vert Y),
\]
where 
\[
\begin{split}
&\ren^{\gamma}_{\alpha}(X \vert Y) = \left (\int_T \ren^{\gamma}_{\alpha}(X \vert Y = y)^{\frac{\alpha - 1 }{\alpha}} \d \prob_{Y} \right)^{\frac{\alpha}{\alpha - 1}} \\ 
%&= \left (\int_T g(y) \ren^{\gamma}_{\alpha}(X \vert Y = y)^{\frac{\alpha - 1 }{\alpha}}  \d \eta(y) \right)^{\frac{\alpha}{\alpha - 1}}. 
\end{split}
\]
\end{Definition}
%\begin{enumerate}
%\item[•]
%\item 
%The \textit{conditional R\'enyi probability} $\ren^{\gamma}_{\alpha}(X \vert Y = y)$ is just the R\'enyi probability of the conditional distribution of $X$, given $Y = y$. Therefore,
%\[
%\ren^{\gamma}_{\alpha}(X \vert Y = y) = \left( \int \left(\frac{F(x,y)}{g(y)}\right)^{\alpha} \d \gamma(x) \right)^{\frac{1}{\alpha - 1}}.
%\]
%%\item dsaf
%%\end{enumerate}
We can re-write $\ren^{\gamma}_{\alpha}(X \vert Y)$ as 
\[
\ren^{\gamma}_{\alpha}(X \vert Y) = \left( \int_{\support (\prob_{Y})} g(y) \left( \int_S \left( \frac{F(x,y)}{g(y)} \right)^{\alpha} \d \gamma(x) \right)^{\frac{1}{\alpha}} \d \eta(y) \right)^{\frac{\alpha}{\alpha - 1}},
\]
which is the expected $L^{\alpha}(S, \gamma)$ norm of the conditional density under the measure $\prob_{Y}$ raised to a power which is the H\"older conjugate of $\alpha$. 
Using Fubini's theorem the formula for $\ren^{\gamma}_{\alpha}(X \vert Y)$ can be further written down only in terms of the joint density,
\[
\ren^{\gamma}_{\alpha}(X \vert Y) = \left( \int_T \left( \int_S F(x,y)^{\alpha} \d \gamma(x) \right)^{\frac{1}{\alpha}} \d \eta(y) \right)^{\frac{\alpha}{\alpha - 1}}.
\]

\begin{Remark} \label{rem: defrem}
Suppose $\prob_{X \vert Y=y}$, for each $y \in \support(g)$, denotes the conditional distribution of $X$ given $Y=y$, i.e. the probability measure on $S$ with density $F(x,y)/g(y)$ with respect to $\gamma$, then the conditional R\'enyi entropy can be written as
\[
h^{\gamma}_{\alpha}(X \vert Y) = \frac{\alpha}{1- \alpha} \log \int_T e^{- \frac{1-\alpha}{\alpha} D_{\alpha}(\prob_{X \vert Y=y} \Vert \gamma)} \d \prob_{Y}(y),
\]
where $D_{\alpha}(\cdot \Vert \cdot)$ denotes R\'enyi divergence (see \ref{def: renyidiv}). 
\end{Remark}

When $X$ and $Y$ are independent random variables one can easily check that $\ren^{\gamma}_{\alpha}(X \vert Y) = \ren^{\gamma}_{\alpha}(X)$, therefore $h^{\gamma}_{\alpha}(X \vert Y) = h^{\gamma}_{\alpha}(X)$ as expected. Since the independence of $X$ and $Y$ means that all the conditionals $\prob_{X \vert Y=y}$ are equal to $\prob_{X}$, the fact that $h^{\gamma}_{\alpha}(X \vert Y) = h^{\gamma}_{\alpha}(X)$ in this case can also be verified from the expression in Remark \ref{rem: defrem}. The converse is also true, i.e. $h^{\gamma}_{\alpha}(X \vert Y) = h^{\gamma}_{\alpha}(X)$ implies the independence of $X$ and $Y$, if $\alpha \neq 0 , \infty$. This is noted later in Corollary \ref{cor:condredentr}.

Clearly, unlike conditional Shannon entropy, the conditional R\'enyi entropy is not the average R\'enyi entropy of the conditional distribution. The average R\'enyi entropy of the conditional distribution, \[ \tilde{h}^{\gamma}_{\alpha}(X\vert Y):= \E_{Y} h^{\gamma}_{\alpha} (X \vert Y=y),\] has been proposed as a candidate for conditional R\'enyi entropy, however it does not satisfy some properties one would expect such a notion to satisfy, like monotonicity (see \cite{FB14}). When $\alpha \geq 1$ it follows from Jensen's inequality that $h_{\alpha}^{\gamma} (X \vert Y) \leq \tilde{h}_{\alpha}^{\gamma} (X \vert Y)$, while the inequality is reversed when $0 < \alpha < 1$.

\section{Relation to R\'enyi Divergence}
\label{sec:div}

We continue to consider an $S$-valued random variable $X$ and a $T$-valued random variable $Y$ with a given joint distribution $\prob_{X,Y}$ with density $F$ with respect to $\gamma \otimes \eta$. Densities, etc are with respect to the fixed reference measures on the state spaces, unless mentioned otherwise. 

Let $\mu$ be a Borel probability measure with density $p$ on a Polish space $\Omega$ and let $\nu$ be a Borel measure with density $q$ on the same space with respect to a common measure $\gamma$.
\begin{Definition}[R\'enyi divergence] \label{def: renyidiv}
Suppose $\alpha \in (0,1) \cup (1, \infty)$. Then, the R\'enyi divergence of order $\alpha$ between measures $\mu$ and $\nu$ is defined as
\[
D_{\alpha} (\mu \Vert \nu) = \log \left( \int_{\Omega} p(x)^{\alpha}  q(x)^{1- \alpha}  \d \gamma (x) \right)^{\frac{1}{\alpha - 1}}.
\]
\end{Definition}
For order $0,1, \infty,$ the R\'enyi divergence is defined by the respective limits.
\begin{Definition}
\begin{enumerate}
\item[]
\item $D_{0}(\mu \Vert \nu) := \lim_{\alpha \to 0} D_{\alpha}(\mu \Vert \nu) = - \log \nu (\support(p))$;
\item $D_{1}(\mu \Vert \nu) := D(\mu \Vert \nu ) = \int p(x) \log \frac{p(x)}{q(x)} \d \gamma(x)$; and 
\item $D_{\infty}(\mu \Vert \nu) := \lim_{\alpha \to \infty} D_{\alpha}(\mu \Vert \nu) = \log \left( \essup_{\mu} \frac{p(x)}{q(x)} \right).$
\end{enumerate}
\end{Definition}
\begin{Remark}
These definitions are independent of the reference measure $\gamma$. 
\end{Remark}
\begin{Remark}
 $\lim_{\alpha \to 1} D_{\alpha}(\mu \Vert \nu)=1$ if some $ D_{\alpha}(\mu \Vert \nu) < \infty$. See \cite{EH14}.
\end{Remark}
%R\'enyi divergence can be used to write a density-free expression for the conditional R\'enyi entropy:
%\[
%h^{\gamma}_{\alpha}(X \vert Y) = \frac{\alpha}{1- \alpha} \log \int e^{- \frac{1-\alpha}{\alpha} D_{\alpha}(W(\cdot , y) \Vert \gamma)} \d \prob_{Y}(y), 
%\]
%in terms of R\'enyi divergence and the transition kernel $W$ corresponding to the conditional distribution. 
%It can be easily checked that the R\'enyi entropies (with respect to the reference measure $\gamma$ of course) of the random variable $X$ can be written in terms of R\'enyi divergences as 
%\[
%h^{\gamma}_{\alpha} (X) = - D_{\alpha} (\prob_{X} \Vert \gamma).
%\] 
The conditional R\'enyi entropy can be written in terms of R\'enyi divergence from the joint distribution using a generalized Sibson's identity we learnt from B. Nakibo\u{g}lu \cite{Nak20:pers} 
(also see \cite{Nak19:1}, and \cite{Sib69} where this identity for $\alpha \neq 1$ appears to originate from). The proof for abstract alphabets presented here is also 
due to B. Nakibo\u{g}lu \cite{Nak20:pers}, which simplifies our original proof \cite{AM19:isit} of the second formula below.

\begin{Theorem}\label{Theorem:conditionalanddivergence}
Let $X,Y$ be random variables taking vaules in spaces $S,T$ respectively. We assume they are jointly distributed with density $F$ with respect to the product reference measure $\gamma \otimes \eta$. For $\alpha \in (0, \infty)$, and any probability measure $\lambda$ absolutely continuous with respect to $\eta$, we have
\[
D_{\alpha} (\prob_{X,Y} \Vert \gamma \otimes \lambda) = D_{\alpha} (\prob_{X,Y} \Vert \gamma \otimes q_{\star}) + D_{\alpha} (q_{\star} \Vert \lambda) = - h^{\gamma}_{\alpha}(X \vert Y) + D_{\alpha}(q_{\star} \Vert \lambda),
\]
where $q_{\star} =\frac{\mu_{\star}}{\Vert \mu_{\star} \Vert}$, $\mu_{\star}$ is the measure having density $ \phi(y)=  \left( \int_{S} F(x,y)^{\alpha} \d \gamma (x) \right)^{\frac{1}{\alpha}}$ with respect to $\eta$, and $\Vert \mu_{\star} \Vert = \int_{T}\left( \int_{S} F(x,y)^{\alpha} \d \gamma (x) \right)^{\frac{1}{\alpha}} \d \eta (y)$ is the normalization factor.

As a consequence, we have
\[
h^{\gamma}_{\alpha}(X \vert Y) = - \min_{\lambda \in \mathcal{P} (T)} D_{\alpha} (\prob_{X,Y} \Vert \gamma \otimes \lambda ),
\]
where $\mathcal{P}(T)$ is the space of probability measures on $T$.
%and the optimal probability measure $\lambda$ in the above, for simple orders, is the one with density proportional to $ \left( \int F(x,y)^{\alpha} \d \gamma (x) \right)^{\frac{1}{\alpha}}$. 
\end{Theorem}
\begin{proof}
Suppose $\lambda$ has density $h$ with respect to $\eta$. Then $\gamma \otimes \lambda$ has density $h(y)$ with respect to $\d \gamma (x) \d \eta (y)$. Now, for $\alpha \neq 1$,
\[
\begin{split}
D_{\alpha}(\prob_{X,Y} \Vert \gamma \otimes \lambda) &= \frac{1}{\alpha - 1} \log \int_{T} \int_{S} F(x,y)^{\alpha} h(y)^{1 - \alpha} \d \gamma (x) \d \eta (y) \\
&= \frac{1}{\alpha - 1 } \log \int_{T} h(y)^{1 - \alpha} \int_S F(x,y)^{\alpha} \d \gamma (x) \d \eta (y) \\
&= \frac{1}{\alpha - 1} \log \int_{T} h(y)^{1 - \alpha}\phi(y)^{\alpha} \d \eta (y) \\
&=\frac{1}{\alpha -1} \log \int_{T} \left( \frac{\d \lambda}{\d \eta} \right)^{1 - \alpha} \left( \frac{\d \mu_{\star}}{\d \eta} \right)^{\alpha} \d \eta (y) \\
&=\frac{1}{\alpha -1} \log \int_{T} \left( \frac{\d \lambda}{\d \eta} \right)^{1 - \alpha} \left( \frac{\d q_{\star}}{\d \eta} \right)^{\alpha} \Vert \mu_{\star} \Vert^{\alpha} \d \eta (y) \\
&= \frac{\alpha}{\alpha - 1} \log\Vert \mu_{\star} \Vert + \frac{1}{\alpha -1} \log \int_{T} \left( \frac{\d \lambda}{\d \eta} \right)^{1 - \alpha} \left( \frac{\d q_{\star}}{\d \eta} \right)^{\alpha} \d \eta (y) \\
& = \frac{\alpha}{\alpha - 1} \log \Vert \mu_{\star} \Vert + D_{\alpha} (q_{\star} \Vert \lambda) \\
&= - h^{\gamma}_{\alpha}(X \vert Y) + D_{\alpha}(q_{\star} \Vert \lambda).
\end{split}
\]
The case $\alpha =1$ is straightforward and well-known, and the optimal $q_{\star}$ in this case is the distribution of $Y$.
\end{proof}

\begin{Remark} 
The identities above and the measure $q_{\star}$ are independent of the reference measure $\eta$. The measure $\eta$ is used only to write out the R\'enyi divergence concretely in terms of densities. 
%Since the $\phi$ in the proof above is supported in the set on which $\nu$ is, one can take the minimum over the set of probabilities measures on $\support ( \nu)$. 
\end{Remark}
%
%Observe that an equivalent way of writing Theorem~\ref{Theorem:conditionalanddivergence} is to say that the conditional R\'enyi entropy is in fact a
%maximum of joint R\'enyi entropies:
%$h^{\gamma}_{\alpha}(X \vert Y) = \max_{\lambda \in \mathcal{P} (T)} h_{\alpha}^{\gamma \otimes \lambda} (X,Y)$.
%the R\'enyi entropy of $X$ of order $\alpha$ is the maximum of R\'enyi entropies of order $\alpha$ over all couplings involving $X$, 
%absolutely continuous with respect to the product measure $\gamma \otimes \eta$. 

\section{Basic properties}
\label{sec:pty}

\subsection{Special orders}

We will now look at some basic properties of the conditional R\'enyi entropy we have defined above. 
First we see that the conditional R\'enyi entropy is consistent with the notion of conditional Shannon entropy of $X$ given $Y$ defined by 
\[
h^{\gamma} (X \vert Y) = - \int_T \int_S F(x,y) \log \frac{F(x,y)}{g(y)} \d \gamma(x) \d \eta(y).
\] 

\begin{Proposition}\label{Proposition:lim1}
\[
\lim_{\alpha \to 1^{+}} h^{\gamma}_{\alpha}(X \vert Y) = h^{\gamma}_{1}(X \vert Y) = h^{\gamma}(X \vert Y),
\]
if $ h^{\gamma}_{\alpha}(X \vert Y)< \infty$ for some $\alpha > 1$.
\end{Proposition}
\begin{proof}
We will use the formula in \ref{Theorem:conditionalanddivergence}. By the monotonicity in the order $\alpha$ of $h_{\alpha}(X \vert Y)$, all limits $\lim_{\alpha \to 1^{+}}, \lim_{\alpha \to 1^{-}}, \lim_{\alpha \to 0} = \lim_{\alpha \to 0^{+}}$ exist. 
Furthermore, for every $\eta$, 
\[ \min_{\lambda} D_{\alpha} (P_{X,Y} \Vert \gamma \otimes \lambda) \leq D_{\alpha} (P_{X,Y} \Vert \gamma \otimes \eta),\] so
\[ \lim_{\alpha \to 1^{+}} \min_{\lambda} D_{\alpha} (P_{X,Y} \Vert \gamma \otimes \lambda) \leq \lim_{\alpha \to 1^{+}} D_{\alpha} (P_{X,Y} \Vert \gamma \otimes \eta),\] that is, 
\[ - \lim_{\alpha \to 1^{+}} h^{\gamma}_{\alpha} (X \vert Y) \leq D (P_{X,Y} \Vert \gamma \otimes \eta).\]  
Now by minimizing over $\eta$ and hitting both sides with a minus sign yields,
\[ \lim_{\alpha \to 1^{+}} h^{\gamma}_{\alpha}(X \vert Y) \geq h^{\gamma} (X \vert Y).\] 

Suppose $\alpha  \geq 1$, then by nondecreasing-ness of the R\'enyi divergence in order, for every $\lambda$ we have
\[ D_{\alpha} (P_{X,Y} \Vert \gamma \otimes \lambda) \geq D (P_{X,Y} \Vert \gamma \otimes \lambda),\] and so by minimization over $\lambda$ and hitting with minus signs we obtain $h^{\gamma}_{\alpha} (X \vert Y) \leq h^{\gamma} (X \vert Y)$. This shows that \[\lim_{\alpha \to 1^{+}} h^{\gamma}_{\alpha} (X \vert Y) = h^{\gamma}(X \vert Y).\]  
\end{proof}

We can extend our definition of R\'enyi probability of order $\alpha$ to $\alpha = 0$ by taking limits, thereby obtaining
\[
\ren^{\gamma}_{0}(X) = \frac{1}{\gamma(\support(f))}.
\]
In the next proposition we define the conditional R\'enyi entropy of order $0$ and record a consequence. 

\begin{Proposition}\label{Proposition:lim0}
\[
h^{\gamma}_{0}(X \vert Y) := \Vert h^{\gamma}_{0}(X \vert Y=y ) \Vert_{L^\infty (\prob_{Y})} = \lim_{\alpha \to 0}h^{\gamma}_{\alpha} (X \vert Y)
\]
\end{Proposition}
\begin{proof}
We will again use the formula from \ref{Theorem:conditionalanddivergence} in this proof. Just as in the last proof, for every probability measure $\eta$,
\[ \min_{\lambda} D_{\alpha} (P_{X,Y} \Vert \gamma \otimes \lambda) \leq D_{\alpha} (P_{X,Y} \Vert \gamma \otimes \eta),\] so
\[ \lim_{\alpha \to 0} \min_{\lambda} D_{\alpha} (P_{X,Y} \Vert \gamma \otimes \lambda) \leq \lim_{\alpha \to 0} D_{\alpha} (P_{X,Y} \Vert \gamma \otimes \eta),\] that is, 
\[ - \lim_{\alpha \to 0} h^{\gamma}_{\alpha} (X \vert Y) \leq D_{0} (P_{X,Y} \Vert \gamma \otimes \eta).\] 
Now by minimizing over $\eta$ and hitting both sides with a minus sign yields,
\[ \lim_{\alpha \to 0} h^{\gamma}_{\alpha}(X \vert Y) \geq h^{\gamma}_{0}(X \vert Y).\] 
Suppose $\alpha  \geq 0$, then by nondecreasing-ness of the R\'enyi divergence in order, for every $\lambda$ we have
\[ D_{\alpha} (P_{X,Y} \Vert \gamma \otimes \lambda) \geq D_{0} (P_{X,Y} \Vert \gamma \otimes \lambda),\] and so by minimization over $\lambda$ and hitting with minus signs we obtain $h^{\gamma}_{\alpha} (X \vert Y) \leq h^{\gamma}_{0} (X \vert Y)$. This shows that \[\lim_{\alpha \to 0} h^{\gamma}_{\alpha} (X \vert Y) = h^{\gamma}_{0}(X \vert Y).\]
\end{proof}

\subsection{Monotonicity in order}

The (unconditional) R\'enyi entropy decreases with order, and the same is true of the conditional version. 

\begin{Proposition}\label{Proposition: conditional renyi is monotone in order}
For random variables $X$ and $Y$, 
\[
h^{\gamma}_{\beta}(X \vert Y) \leq h^{\gamma}_{\alpha} (X \vert Y),
\]
for all $0 < \alpha \leq \beta \leq \infty$.
\end{Proposition}

The proof is essentially the same as in the discrete setting, and follows from Jensen's inequality.

\begin{proof}
We obtain separately the cases $1 < \alpha \leq \beta < \infty$, $0 < \alpha \leq \beta <1$, so that the entire Proposition follows 
from taking limits and the transitivity of inequality. Let $\alpha \leq \beta$ be positive numbers, set $e= \frac{\beta - 1}{\beta} \frac{\alpha}{\alpha - 1}$. Consider the following argument. 
\[
\begin{split}
\ren_{\beta}(X \vert Y) &= \left( \int_{\support(\prob_{Y})} \ren^{\gamma}_{\beta} (X \vert Y=y)^{\frac{\beta - 1}{\beta}}  \d \prob_{Y} \right)^{\frac{\beta}{\beta - 1}} \\
& \geq \left( \int_{\support(\prob_{Y})} \ren^{\gamma}_{\alpha} (X \vert Y=y)^{\frac{\beta - 1}{\beta}}  \d \prob_{Y} \right)^{\frac{\beta}{\beta - 1}} \\
&= \left( \int_{\support(\prob_{Y})} \ren^{\gamma}_{\alpha} (X \vert Y=y)^{\frac{\alpha - 1}{\alpha}e}  \d \prob_{Y} \right)^{\frac{\beta}{\beta - 1}} \\
& \geq \left( \int_{\support(\prob_{Y})} \ren^{\gamma}_{\alpha} (X \vert Y=y)^{\frac{\alpha - 1}{\alpha}}  \d \prob_{Y} \right)^{e \frac{\beta}{\beta - 1}} \\
&= \ren^{\gamma}_{\alpha}(X \vert Y).
\end{split}
\]
The first inequality above follows from the fact that the unconditional R\'enyi entropy (probability) decreases (increases) with order. Note that $e \geq 1$ when $1 < \alpha \leq \beta< \infty$ and hence the function $r \mapsto r^{e}$ is convex making the second inequality an application of Jensen's inequality in this case. When $0< \alpha \leq \beta <1$, the exponent satisfies $0 < e \leq 1$ so the function $r \mapsto r^{e}$ is concave but the outer exponent $\frac{\beta}{\beta -1}$ is negative which turns the (second) inequality in the desired direction.  
\end{proof}

\subsection{Conditioning reduces R\'enyi entropy}

As is the case for Shannon entropy, we find that the conditional R\'enyi entropy obeys monotonicity too; 
the proof of the theorem below adapts the approach for the discrete case taken in \cite{FB14} by using Minkowski's integral inequality.

\begin{Theorem}\label{thm:mono}[Monotonicity]
Let $\alpha \in [0, \infty]$, and $X$ be $S$-valued and $Y,Z$ be $T$-valued random variables. Then,
\[
h^{\gamma}_{\alpha}(X \vert YZ) \leq h^{\gamma}_{\alpha} (X \vert Z).
\]
\end{Theorem}

\begin{proof}
We begin by proving the result for an empty $Z$.

First, we deal with the case $ 1 < \alpha < \infty$. In terms of R\'enyi probabilities we must show that conditioning increases R\'enyi probability. Indeed, 
\begin{eqnarray}\label{eq:ren-prob-mink}
\begin{split}
\ren^{\gamma}_{\alpha}(X \vert Y ) 
&= \left( \int_T
   \left[ 
   \int_S F(x,y)^{\alpha} \d \gamma (x)
   \right]^{\frac{1}{\alpha}} \d \eta(y)
 \right)^{\frac{\alpha}{\alpha - 1}}
 \\
& \geq 
\left[ \int_S
 \left( \int_T F(x,y) \d \eta (y) \right)^{\alpha} \d \gamma (x)
\right]^{\frac{1}{\alpha} \frac{\alpha}{\alpha - 1 }} \\ 
&= \left( \int_S f(x)^{\alpha} \d \gamma(x) \right)^{\frac{1}{\alpha - 1}} = \ren^{\gamma}_{\alpha} (X).
\end{split}
\end{eqnarray}
The inequality above is a direct application of Minkowski's integral inequality \cite[Theorem 2.4]{LL01:book}, 
which generalizes the summation in the standard triangle inequality to integrals against a measure.

For the case $0 < \alpha < 1$, we apply the triangle inequality for $1 / \alpha$ then the fact that now $\frac{1}{\alpha - 1}$ is negative flips the inequality in the desired direction:
\[
\begin{split}
\ren^{\gamma}_{\alpha}(X \vert Y )
&= \left( \int_T
   \left[ 
   \int_S F(x,y)^{\alpha} \d \gamma (x)
   \right]^{\frac{1}{\alpha}} \d \eta(y)
 \right)^{\frac{\alpha}{\alpha - 1}}
 \\ 
 & \geq 
\left[ \int_S
 \left( \int_T F(x,y)^{\frac{\alpha}{\alpha}} \d \eta (y) \right)^{\alpha} \d \gamma (x)
\right]^{ \frac{1}{\alpha - 1 }} \\ 
&= \left( \int_S f(x)^{\alpha} \d \gamma(x) \right)^{\frac{1}{\alpha - 1}} = \ren^{\gamma}_{\alpha} (X).
\end{split}
\] 

Now to extend this to non-empty $Z$ we observe the following. 
Suppose $V$ is an $S$-valued random variable, $W$ is a $T$-valued random variable and $ h \in \R$ is such that
\[
h^{\gamma}_{\alpha}(V \vert W, Z = z) \geq h^{\gamma}_{\alpha}(X \vert Y, Z = z) + h,
\]
for every $z$ in the support of $\prob_{Z}$. Then 
\[
h^{\gamma}_{\alpha}(V \vert WZ) \geq h^{\gamma}_{\alpha}(X \vert Y Z) + h.
\]
In terms of Renyi probabilities this means that if for every $z \in \support (\prob_{Z})$, \[ \ren^{\gamma}_{\alpha}(V\vert W , Z = z)  \leq \frac{\ren^{\gamma}_{\alpha} (X \vert  Y, Z = z)}{\log h}, \] then 
\[
\ren^{\gamma}_{\alpha}(V \vert W Z) \leq \frac{\ren^{\gamma}_{\alpha} (X \vert YZ)}{\log h }.
\]
Indeed, 
\[
\begin{split}
\ren^{\gamma}_{\alpha}(V \vert W Z) 
&= \left( \int_{\support (Z)} \ren^{\gamma}_{\alpha}(V \vert W , Z = z)^\frac{\alpha - 1}{\alpha} \d \prob_{Z}(z)  \right) ^{\frac{\alpha}{\alpha - 1}} \\ & \leq 
\left( \int_{\support (Z)} \left(  \frac{\ren^{\gamma}_{\alpha} (X \vert  Y, Z = z)}{\log h}  \right)^\frac{\alpha - 1}{\alpha} \d \prob_{Z}(z)  \right) ^{\frac{\alpha}{\alpha - 1}} \\
&= \frac{\ren^{\gamma}_{\alpha} (X \vert YZ)}{\log h },
\end{split}
\]
since the functions $r \mapsto r^{\frac{\alpha - 1}{\alpha}}$ and $r \mapsto r^{\frac{\alpha}{\alpha - 1}}$ are both strictly increasing or strictly decreasing, based on the value of $\alpha$.
Finally, the claim for non-empty $Z$ follows from this observation given we have already demonstrated $h^{\gamma}_{\alpha}(X \vert Y , Z = z) \leq h^{\gamma}_{\alpha}(X \vert Z = z)$ throughout $\support (Z)$. 
The cases $\alpha = 0 , \infty$ are obtained by taking limits. For $\alpha = 1$ this is well-known.
\end{proof}
\begin{Corollary}
If $ X \to Y \to Z$ forms a Markov chain, then $h^{\gamma}_{\alpha}(X \vert Y) \leq h^{\gamma}_{\alpha} (X \vert Z)$. 
\end{Corollary}

When we specialize to ``empty $Z$'' (i.e., the $\sigma$-field generated by $Z$ being the Borel $\sigma$-field on $T$, or not conditioning
on anything),  we find that ``conditioning reduces R\'enyi entropy''. 

\begin{Corollary}\label{cor:condredentr}
Let $\alpha \in (0,\infty)$. Then $h^{\gamma}_{\alpha}(X \vert Y) \leq h^{\gamma}_{\alpha} (X)$, with equality \textit{iff} $X$ and $Y$ are independent. 
\end{Corollary}

While the inequality in Corollary~\ref{cor:condredentr} follows immediately from Theorem~\ref{thm:mono}, the conditions for
equality follow from those for Minkowski's inequality (which is the key inequality used in the proof of Theorem~\ref{thm:mono},
see, e.g., \cite[Theorem 2.4]{LL01:book}): 
given the finiteness of both sides in the display \eqref{eq:ren-prob-mink}, equality holds if and only if 
$F(x,y)= \phi(x) \psi(y)$ $\gamma \otimes \eta$-a.e. for some functions $\phi$ and $\psi$. In our case, this means that 
equality holds in $h^{\gamma}_{\alpha}(X \vert Y) \leq h^{\gamma}_{\alpha} (X)$ if and only if $X$ and $Y$ are independent ($\alpha \in (0,1)\cup (1, \infty)$). The corresponding statement for $\alpha = 1$ is well-known. 

Since $\tilde{h}_{\alpha}^{\gamma} (X \vert Y) \leq h_{\alpha}^{\gamma} (X \vert Y)$ when $0 < \alpha < 1$, 
as noted in Section \ref{sec:defn}, we have ``conditioning reduces R\'enyi entropy'' in this case as well. 

\begin{Corollary}
$\tilde{h}_{\alpha}^{\gamma} (X \vert Y) \leq h^{\gamma}_{\alpha}(X)$, when $0 < \alpha < 1$. 
\end{Corollary}
\begin{Remark}
The above corollary is not true for large values of $\alpha$. For a counter-example, see \cite[Theorem 7]{TMA12}.
\end{Remark}

From the special case when $Y$ is discrete random variable taking finitely many values $y_{i}$ with probability $p_{i}$, $1 \leq i \leq n$, 
and the conditional density of $X$ given $Y=y_{i}$ is $f_{i}(x)$, we obtain the concavity of R\'enyi entropy (for orders below 1) 
which is already known in the literature.

\begin{Corollary} Let $0 \leq \alpha \leq 1$. Suppose $f_{i}$ are densities on $S$ and $p_{i}$ non-negative numbers, $1 \leq i \leq n$, such that $\sum_{i}p_{i} = 1$. Then,
$h^{\gamma}_{\alpha}(\sum_{i=1}^{n} p_{i}f_{i}) \geq 
 \sum_{i=1}^{n} p_{i} h^{\gamma}_{\alpha}(f_{i})$.
\end{Corollary} 

\subsection{A chain rule}
In this subsection we deduce a version of the chain rule from Theorem \ref{Theorem:conditionalanddivergence}. For the discrete case, 
this has been done by Fehr and Berens in \cite[Theorem 3]{FB14}.
If $\eta$ is a probability measure, then we have 
$h_{\alpha}^{\gamma}(X \vert Y) \geq - D_{\alpha}(\prob_{X,Y} \Vert \gamma \otimes \eta) = h_{\alpha}^{\gamma \otimes \eta}(X,Y)$. 
If we relax the condition on $\eta$ to be a measure under which $\prob_{Y}$ is absolutely continuous and supported on a set of finite measure, we obtain 
$h_{\alpha}^{\gamma}(X \vert Y) \geq h_{\alpha}^{\gamma \otimes \eta}(X,Y) - h_{0}^{\eta}(Y)$. 
Since this inequality trivially holds true when $Y$ is supported on a set of infinite $\eta$-measure, we have proved
the following inequality that (although weaker being an inequality rather an identity) is reminiscent of the chain rule for Shannon entropy.
%As observed in the discrete setting by \cite{FB14}, a version of the chain rule does persist
%for R\'enyi entropy, but is much weaker than in the Shannon case for two reasons: it is an inequality
%rather than an identity, and instead of all terms being R\'enyi information measures of order $\alpha$,
%one of the terms is an information measure of order 0.
 
\begin{Proposition}\label{Proposition:chain}
For any $\alpha>0$,
\[
h^{\gamma \otimes \eta}_{\alpha}(X,Y)\leq h^{\gamma}_{\alpha} (X \vert Y)  + h^{\eta}_{0}(Y) .
\]
\end{Proposition}

\begin{proof}
Recall that
\[
\ren^{\gamma}_{\alpha}(X \vert Y) = \left[ \int_T \left( \int_S F(x,y)^{\alpha} \d \gamma (x) \right)^{\frac{1}{\alpha}} \d \eta (y) \right]^{\frac{\alpha}{\alpha - 1}},
\]
where the outer integral can be restricted to the support of $\prob_{Y}$, which we will keep emphasizing in the first few steps.
\[
\begin{split}
\ren^{\gamma}_{\alpha}(X \vert Y) 
&= \left[ \int_{\support(\prob_{Y})} \left( \int_S F(x,y)^{\alpha} \d \gamma (x) \right)^{\frac{1}{\alpha}} \d \eta (y) \right]^{\frac{\alpha}{\alpha - 1}} \\
&= \left[ \int_{\support(\prob_{Y})} \left( \eta (\support(\prob_{Y})) \int_S F(x,y)^{\alpha} \d \gamma (x) \right)^{\frac{1}{\alpha}} \d \frac{\eta (y)}{\eta(\support(\prob_{Y}))} \right]^{\frac{\alpha}{\alpha - 1}} \\
&= \left[ \int_{\support(\prob_{Y})} \left(  \int_S \left( \eta (\support(\prob_{Y}))F(x,y) \right)^{\alpha} \d \gamma (x) \right)^{\frac{1}{\alpha}} \d \frac{\eta (y)}{\eta(\support(\prob_{Y}))} \right]^{\frac{\alpha}{\alpha - 1}} .
\end{split}
\] 
By Jensen's inequality, when $\alpha > 1$,
\[
\begin{split}
\ren^{\gamma}_{\alpha}(X \vert Y) 
& \leq \left[ \int_{\support(\prob_{Y})} \left(  \int_S \left( \eta (\support(\prob_{Y}))F(x,y) \right)^{\alpha} \d \gamma (x) \right) \d \frac{\eta (y)}{\eta(\support(\prob_{Y}))} \right]^{\frac{1}{\alpha - 1}} \\
&= \eta(\support (\prob_{Y})) \left[ \int_{\support(\prob_{Y})} \int_S F(x,y)^{\alpha} \d \gamma (x) \d \eta (y)  \right]^{\frac{1}{\alpha - 1}}
\\
& = \eta (\support(\prob_{Y})) \ren^{\gamma \otimes \eta}_{\alpha}(X,Y).
\end{split}
\] 
Note that the above calculation also holds when $\alpha \in (0,1)$ because even though Jensen's inequality is flipped because $\frac{1}{\alpha}$ is now convex, the inequality flips again, now to the desired side, since the exponent $\frac{\alpha}{\alpha - 1 }$ is negative. Taking logarithms and hitting both sides with a minus sign concludes the proof. 
\end{proof}
\begin{Remark}
These inequalities are tight. Equality is attained when $X$, $Y$ are independent and $Y$ is uniformly distributed on finite support. 
\end{Remark}
%Since probabilities are bounded by 1, we immediately have that $h^{\gamma}_\alpha(X|Y)\geq h^{\gamma \otimes \eta}_\alpha(X,Y)$ when $\eta$ is a probability measure;
%this can also be seen as a consequence of the final Remark in Section~\ref{sec:div}.

\subsection{Sensitivity to reference measure}

\begin{Proposition}\label{Proposition:compare}
Suppose $\frac{\d \gamma}{\d \mu} = \psi$ which is bounded by a number $M$. Then $
h^{\mu}_{\alpha}(X \vert Y) \geq  h^{\gamma}_{\alpha}(X \vert Y) - \log M.$
\end{Proposition}

\begin{proof}
 Then the joint density of $(X,Y)$ under the measure $\mu \otimes \mu$ becomes $F(x,y)\psi(x) \psi(y)$ if the joint density was $F(x,y)$ under $\gamma$. Now suppose $\alpha \geq 1$ then, 
\[
\begin{split}
&h^{\mu}_{\alpha}(X \vert Y) \\
&= \frac{- \alpha}{\alpha - 1} \log \left[  \int_T \left( \int_S F(x,y)^{\alpha} \psi(x)^{\alpha} \psi(y)^{\alpha} \d \mu(x) \right)^{\frac{1}{\alpha}}  \d \mu (y)  \right] \\
&= \frac{- \alpha}{\alpha - 1} \log \left[  \int_T \left( \int_S F(x,y)^{\alpha} \psi(x)^{\alpha}  \d \mu(x) \right)^{\frac{1}{\alpha}} \psi(y) \d \mu (y)  \right] \\
&=  \frac{- \alpha}{\alpha - 1} \log \left[  \int_T \left( \int_S F(x,y)^{\alpha} \psi(x)^{\alpha - 1}  \d \gamma (x) \right)^{\frac{1}{\alpha}}  \d \gamma (y)  \right] \\
& \geq \frac{- \alpha}{\alpha - 1} \left( \log \left[  \int_T \left( \int_S F(x,y)^{\alpha}  \d \gamma (x) \right)^{\frac{1}{\alpha}}  \d \gamma (y)  \right] + \log M^{\frac{\alpha - 1 }{\alpha}} \right) \\
&= h^{\gamma}_{\alpha}(X \vert Y) - \log M.
\end{split}
\]
If $\alpha \in (0,1)$, the same inequality holds if $\mu$ is also absolutely continuous w.r.t. $\gamma$. 
\end{proof}

\section{Notions of $\alpha$-mutual information}
\label{sec:mi}
Arimoto \cite{Ari77} used his conditional R\'enyi entropy to define a mutual information that we extend to the
general setting as follows.
\begin{Definition}
Let $X$ be an $S$-valued random variable and let $Y$ be a $T$-valued random variable, with a given joint distribution. Then, we define
\[
I^{(\gamma)}_{\alpha} (X\leadsto Y) = h^{\gamma}_{\alpha}(X) - h^{\gamma}_{\alpha}(X \vert Y).
\]
\end{Definition}
We use the squiggly arrow to emphasize the lack of symmetry in $X$ and $Y$, but nonetheless to distinguish from the notation for directed mutual
information, which is usually written with a straight arrow. By Corollary~\ref{cor:condredentr}, for $\alpha \in (0,\infty)$,  $I^{(\gamma)}_{\alpha} (X\leadsto Y) = 0$ if and only if
$X$ and $Y$ are independent. Therefore $I^{(\gamma)}_{\alpha} (X\leadsto Y)$, for any choice of reference measure $\gamma$, can be seen as a measure of dependence between $X$ and $Y$. 

Let us discuss a little further the validity of $I^{(\gamma)}_{\alpha} (X\leadsto Y)$ as a dependence measure.
If the conditional distributions are denoted by $\prob_{X \vert Y=y}$ as in Remark \ref{rem: defrem}, 
using the fact that $h^{\nu}_{\alpha}(Z) = - D_{\alpha}(Z \Vert \nu)$ for any random variable $Z$, we have for any  $\alpha \in (0,1) \cup (1, \infty)$ that
\[
\begin{split}
I^{(\gamma)}_{\alpha} (X\leadsto Y)
&=h^{\gamma}_{\alpha}(X) - h^{\gamma}_{\alpha}(X \vert Y) \\
&= \frac{\alpha}{1- \alpha} \log e^{- \frac{1-\alpha}{\alpha} D_{\alpha}(\prob_{X} \Vert \gamma)}  
- \frac{\alpha}{1- \alpha} \log \int_T e^{- \frac{1-\alpha}{\alpha} D_{\alpha}(\prob_{X \vert Y=y} \Vert \gamma)} \d \prob_{Y}(y) \\
&= - \frac{\alpha}{1- \alpha} \log \int_T e^{- \frac{1-\alpha}{\alpha} [ D_{\alpha}(\prob_{X \vert Y=y} \Vert \gamma) - D_{\alpha}(\prob_{X} \Vert \gamma)]} \d \prob_{Y}(y) .
\end{split}
\]
Furthermore, when $\alpha \in (0,1)$, by \cite[Proposition 2]{EH14}, we may also write 
\[
\begin{split}
I^{(\gamma)}_{\alpha} (X\leadsto Y) &=
 - \frac{\alpha}{1- \alpha} \log \int_T e^{-  [ D_{1 - \alpha}(\gamma \Vert \prob_{ X \vert Y=y} ) - D_{1 - \alpha}(\gamma \Vert \prob_{X}) ]} \d \prob_{Y}(y).
\end{split}
\] 
Note that R\'enyi divergence is convex in the second argument (see \cite[Theorem 12]{EH14}) when $\alpha \in (0,1)$, and the last equation suggests that 
Arimoto's mutual information can be seen as a quantification of this convexity gap.

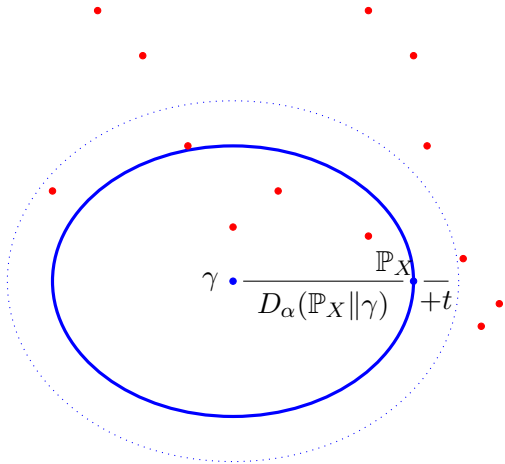
\begin{figure}
\centering
    \begin{tikzpicture}[help lines/.style={blue!30,very thin},scale=0.6]
    [
    every edge quotes/.append style={sloped, font=\sffamily\scriptsize, auto}
  ]
    \draw[color=blue,very thick] (0, 0) ellipse (4cm and 3cm);
    \draw[color=blue,dotted] (0, 0) ellipse (5cm and 4cm);

%    \foreach \x/\y in {0/0, 5/0, -5/0,0/4, 0/-4 , 4/-2.4}
%    \filldraw[black] (\x, \y) circle(1.5pt);

    \foreach \x/\y in { -0/0 , 4/0}
    \filldraw[blue] (\x, \y) circle(2pt);
 \draw[color=black] (0,0) node (A) {};
 \draw[color=blue] (4,0) node (B){};
 \draw[color=white] (5,0) node (C){};
    \node at (-0.5,0) {$\gamma$};
    \node at (3.6,0.4) {$\prob_{X}$};
    \draw (A) edge ["$D_{\alpha}(\prob_{X} \Vert \gamma)$"'] (B);
    \draw (B) edge ["$+t$"'] (C);
      \foreach \x/\y in { 1/2, 4.3/3, 4/5, 3/6 , -1/3, 0/1. 2/-4 , 3/1 , 5.1/0.5, 5.9/-0.5, 5.5/-1, -2/5, -3/6, -4/2}
    \filldraw[red] (\x, \y) circle(2pt);
    \end{tikzpicture}
    \caption{A schematic diagram showing how large $I^{(\gamma)}_{\alpha} (X\leadsto Y)$, for a fixed $\alpha < 1$, demonstrates strong dependence between $X$ and $Y$:
    the space depicted is the space of probability measures on $S$, including $\gamma$, $\prob_{X}$, and the red dots 
    representing the conditional distributions of $X$ given that $Y$ takes different values in $T$. The $D_\alpha$-balls around $\gamma$
    are represented by ellipses to emphasize that the geometry is non-Euclidean and in fact, non-metric.
    When $I^{(\gamma)}_{\alpha} (X\leadsto Y)$ is large, there is a significant probability that $Y$ takes values such that the
    corresponding conditional distributions of $X$ lie outside the larger $D_\alpha$-ball, and therefore far from the
    (unconditional) distribution $\prob_{X}$ of $X$.} \label{fig: one}
    \end{figure}

One can also see clearly from the above expressions why this quantity controls, at least for $\alpha\in (0,1)$, the dependence between $X$ and $Y$: indeed,
one has for any $\alpha \in (0,1)$ and any  $t>0$ that,
\[
\begin{split}
\prob_{Y} \{ D_{\alpha}(\prob_{X \vert Y} \Vert \gamma) - D_{\alpha}(\prob_{X} \Vert \gamma) < t \} 
&=\prob_{Y} \{ e^{-\beta [ D_{\alpha}(\prob_{X \vert Y} \Vert \gamma) - D_{\alpha}(\prob_{X} \Vert \gamma)]} > e^{-\beta t} \} 
\leq   e^{\beta t} e^{-\beta I^{(\gamma)}_{\alpha} (X\leadsto Y)} ,
\end{split}
\]
where the inequality comes from Markov's inequality, and we use $\beta=\frac{1-\alpha}{\alpha}$. Thus, when $I^{(\gamma)}_{\alpha} (X\leadsto Y)$ is large,
the probability that the conditional distributions of $X$ given $Y$ cluster at around the same ``R\'enyi divergence'' distance from the reference measure $\gamma$ 
as the unconditional distribution of $X$ (which is of course a mixture of the  conditional distributions) is small, suggesting a significant ``spread'' of the conditional 
distributions and therefore strong dependence. This is illustrated in Figure~\ref{fig: one}. Thus, 
despite the dependence of $I^{(\gamma)}_{\alpha} (X\leadsto Y)$ on the reference measure $\gamma$,
it does guarantee strong dependence when it is large (at least for $\alpha<1$). 
When $\alpha \to 1^{-}$ we have $\beta \to 0$, and consequently the upper bound $e^{\beta t} e^{-\beta I^{(\gamma)}_{\alpha} (X\leadsto Y)} \to 1$ making the inequality trivial.

The ``mutual information'' quantity $I^{(\gamma)}_{\alpha} (X\leadsto Y)$ clearly depends on the choice of the reference measure $\gamma$.
%However Proposition~\ref{Proposition:compare} allows one to relate the Arimoto mutual informations for different reference measures
%when the likelihood ratio between them is bounded in both directions (i.e., away from 0, and from above).
Nonetheless, there are 3 families of R\'enyi mutual informations that are independent of the choice of reference measure,
which we now introduce.

%When $\alpha = 1$,
%\[
%\prob_{Y} \{ D(\prob_{X \vert Y} \Vert \gamma) - D(\prob_{X} \Vert \gamma) \geq t \} = \prob_{Y} \{ e^{[D(\prob_{X \vert Y} \Vert \gamma) - D(\prob_{X} \Vert \gamma)]} \geq e^t \} \leq e^{-t} \E e^{[D(\prob_{X \vert Y} \Vert \gamma) - D(\prob_{X} \Vert \gamma)]}
%\]
\begin{Definition} Fix $ \alpha \geq 0$.
\begin{enumerate}
\item[]
\item The Lapidoth-Pfister $\alpha$-mutual information is defined as 
\[
J_{\alpha}(X;Y) := \min_{\mu \in \mathcal{P}(S), \nu \in \mathcal{P}(T)} D_{\alpha} (\prob_{X,Y} \Vert \mu \otimes \nu).
\]
\item The Augustin-Csisz\'ar $\alpha$-mutual information is defined as
\[
K_{\alpha}^{X\leadsto Y}(\prob_{X,Y}) = \min_{\mu \in \mathcal{P}(T)} \E_{X} D_{\alpha}(\prob_{Y \vert X}(\cdot \vert X) \Vert \mu).
\]
\item Sibson's $\alpha$-mutual information is defined as
\[
I_{\alpha}^{X\leadsto Y}(\prob_{X,Y}) = \min_{\mu \in \mathcal{P}(T)} D_{\alpha} (\prob_{X,Y} \Vert \prob_{X} \otimes \mu).
\]
\end{enumerate}
\end{Definition}

The quantity $J_{\alpha}$ was recently introduced by Lapidoth and Pfister as a measure of independence in \cite{LP16:icsee} (cf., \cite{LP18:itw, TH18:1, LP19}). The Augustin-Csisz\'ar mutual information was originally introduced in \cite{Aug78:hab} by Udo Augustin with a slightly different parametrization, and gained much popularity following Csisz\'ar's work in \cite{Csi95}. For a discussion on early work on this quantity and applications also see \cite{Nak19:2} and references therein. Both \cite{Aug78:hab} and \cite{Nak19:2} treat abstract alphabets however the former is limited to $\alpha \in (0,1)$ while the latter treats all $\alpha \in (0, \infty)$. Sibson's definition originates in \cite{Sib69} where he introduces $I_{\alpha}$ in the form of \textit{information radius} (see, e.g,  \cite{SV18:1}), which is often written in terms of Gallager's function (from \cite{Gal65}). 
Since all the quantities in the above definition are stated in terms of R\'enyi divergences not involving the reference measure $\gamma$, 
they themselves are independent of the reference measure. Their relationship with the R\'enyi divergence also shows that all of them are non-negative. Moreover, putting $\mu= \prob_{X}, \nu = \prob_{Y}$ in the expression for $J_{\alpha}$ and $\mu = \prob_{Y}$ in expressions for $K_{\alpha}$ and $I_{\alpha}$ when $X,Y$ are independent shows that they all vanish under independence.

While these notions of mutual information are certainly not equal to $I^{(\gamma)}_{\alpha}$ in general when $\alpha \neq 1$, they do have a 
direct relationship with conditional R\'enyi entropies, by varying the reference measure. 

Since 
$\min_{\mu} \min_{\nu} D_{\alpha}(\prob_{X,Y} \Vert \mu \otimes \nu) = \min_{\mu} - h^{\mu}_{\alpha}(X \vert Y) = - \max_{\mu} h^{\mu}_{\alpha}(X \vert Y)$, 
where all optimizations are done over probability measures, we can write Lapidoth and Pfister's mutual information as
\[
J_{\alpha}(X;Y) = - \max_{\mu} h^{\mu}_{\alpha}(X \vert Y) = \min_{\mu} - \frac{\alpha}{1- \alpha} \log \int_T e^{- \frac{1-\alpha}{\alpha} D_{\alpha}(\prob_{X \vert Y=y} \Vert \mu)} \d \prob_{Y}(y).
\]
Note that it is symmetric by definition: $J_{\alpha}(X;Y) =J_{\alpha}(Y;X)$, which is why we do not use squiggly arrows to denote it. 
By writing down R\'enyi divergence as R\'enyi entropy w.r.t. reference measure, Augustin-Csisz\'ar's $K_{\alpha}$ can be recast in a similar form, this time using the average R\'enyi entropy of the conditionals instead of Arimoto's conditional R\'enyi entropy,
\[
K_{\alpha}^{X\leadsto Y}= - \max_{\mu} \E_{X} h^{\mu}_{\alpha} (Y \vert X=x) = - \max_{\mu} \tilde{h}^{\mu}_{\alpha}(Y \vert X).
\] 
%Among different proposed notions of conditional R\'enyi entropy, the one attributed to C. Cachin in \cite{FB14} is
%\[
%\tilde{h}_{\alpha}(Y \vert X) = \E_{X} h_{\alpha}(Y \vert X=x),
%\]
%so upon introducing an abstract reference measure in the most obvious way in the above, one can write Csisz\'ar's mutual information as
%\[
%I^{C}_{\alpha}(X ; Y) = - \max_{\mu} \tilde{h}^{\mu}_{\alpha}(Y \vert X).
%\]

In light of Theorem~\ref{Theorem:conditionalanddivergence}, Sibson's mutual information can clearly be written in terms of conditional R\'enyi entropy as
\[
I_{\alpha}^{X\leadsto Y}= - h^{\prob_{X}}_{\alpha} (X \vert Y).
\]
This leads to the observation that Sibson's mutual information can be seen as a special case of Arimoto's mutual information, when the reference measure is taken to be the distribution of $X$:
\[
\begin{split}
&I_{\alpha}^{X\leadsto Y}= - h^{\prob_{X}}_{\alpha} (X \vert Y)\\
&= h^{\prob_{X}}_{\alpha}(X) - h^{\prob_{X}}_{\alpha}(X \vert Y)= I^{(\prob_{X})}_{\alpha }(X \leadsto Y).
\end{split}
\]
For the sake of comparision with the corresponding expression for $I^{(\gamma)}_{\alpha} (X\leadsto Y)$, we also write $I_{\alpha}^{X\leadsto Y}$ as
\[
I_{\alpha}^{X\leadsto Y} = - \frac{\alpha}{1- \alpha} \log \int_T e^{- \frac{1-\alpha}{\alpha}  D_{\alpha}(\prob_{X \vert Y=y} \Vert \prob_X)} \d \prob_{Y}(y).
\]

The following inequality, which relates the three families when $\alpha \geq 1$, turns out to be quite fruitful.

\begin{Theorem}\label{Theorem:sandwich}
For $\alpha \geq 1$, we have
\[
K_{\alpha}^{X\leadsto Y} \leq J_{\alpha} (X ; Y) \leq I_{\alpha}^{X\leadsto Y}.
\]
\end{Theorem}

\begin{proof}
Suppose $\alpha \geq 1$. Then, $h^{\mu}_{\alpha}(Y \vert X) \leq \tilde{h}^{\mu}_{\alpha}(Y \vert X)$
%\[
%\begin{split}
%&h^{\mu}_{\alpha}(Y \vert X) \\
%&= - \frac{\alpha}{\alpha - 1} \log \int \ren^{\mu}_{\alpha} (Y \vert X=x)^{\frac{\alpha - 1}{\alpha}} \d \prob_{X} \\
%& \leq - \frac{\alpha}{\alpha - 1} \int \frac{\alpha - 1}{\alpha} \log \ren^{\mu}_{\alpha}(Y \vert X = x ) \d \prob_{Y} = \tilde{h}^{\mu}_{\alpha}(Y \vert X), 
%\end{split}
%\]
so that 
\[ 
\begin{split}
K_{\alpha}^{X\leadsto Y} = - \max_{\mu}\tilde{h}^{\mu}_{\alpha} (Y \vert X) &\leq - \max_{\mu} h^{\mu}_{\alpha} (Y \vert X) \\ &= J_{\alpha}(Y;X) = J_{\alpha}(X;Y).
\end{split}
\]
Moreover,
\[
J_{\alpha}(X;Y) = - \max_{\mu} h^{\mu}_{\alpha}(X \vert Y) \leq - h^{\prob_{X}}_{\alpha} (X \vert Y) = I_{\alpha}^{X\leadsto Y},
\]
which completes the proof.
\end{proof}
\begin{Remark}
When $\alpha \in (0,1)$, from the straightforward observation $J_{\alpha}(X ; Y) \leq I^{X \leadsto Y}_{\alpha}$ and \cite{Csi95}, we have $ J_{\alpha} (X ; Y) \leq I_{\alpha}^{X\leadsto Y} \leq K_{\alpha}^{X\leadsto Y}$.  
\end{Remark}
We note that the relation between $K_{\alpha}^{X\leadsto Y}$ and $I_{\alpha}^{X\leadsto Y}$ (in the finite alphabet case)
goes back to Csisz\'ar \cite{Csi95}.
In the next and final section, we explore the implications of Theorem~\ref{Theorem:sandwich} for various notions of capacity.

\section{Channels and capacities}
\label{sec:cc}
We begin by defining channel and capacities. Throughout this section, assume $\alpha \geq 1$.

\begin{Definition}
Let $(A,\mathcal{A}), (B,\mathcal{B})$ be measurable spaces. A function $W : A \times \mathcal{B} \to \R$ is called a probability kernel or a channel from the input space $(A, \mathcal{A})$ to the output space $(B,\mathcal{B})$ if 
\begin{enumerate}
\item For all $a \in A$, the function $W(\cdot \vert a): \mathcal{B} \to \R$ is a probability measure on $(B, \mathcal{B})$, and 
\item For every $V \in \mathcal{B}$, the function $W(V \vert \cdot): A \to \R$ is a $\mathcal{A}$-measurable function. 
\end{enumerate} 
\end{Definition}

%Given a channel $W$ from $(A, \mathcal{A})$ to $(B, \mathcal{B})$ and an input probability measure $p$ on $(A, \mathcal{A})$, there is a unique probability measure $p \star W$ on $(A \times B, \mathcal{A} \otimes \mathcal{B})$ such  that 
%\[
%p \star W (U \times V) = \int_{U} W(V \vert x) \d p(x),
%\]
%for all $U \in \mathcal{A}, V \in \mathcal{B}$. The output measure $pW$ on $(B, \mathcal{B})$ is given by $pW (V) = p \star W (A \times V)$. 
%For example, 
In our setting, the conditional distributions $X$ given $Y=y$ define a channel $W$ from $\support(\prob_{Y})$ to $S$. In terms of this, 
one can write a density-free expression for the conditional R\'enyi entropy:
\[
h^{\gamma}_{\alpha}(X \vert Y) = \frac{\alpha}{1- \alpha} \log \int_T e^{- \frac{1-\alpha}{\alpha} D_{\alpha}(W(\cdot \vert y) \Vert \gamma)} \d \prob_{Y}(y).
\]

\begin{Definition}
Let $(B, \mathcal{B})$ be a measurable space and $\mathcal{W} \subseteq \mathcal{P}(B)$ a set of probability measures on $B$. 
Following \cite{Nak19:1}, define the order-$\alpha$ R\'enyi radius of $\mathcal{W}$ relative to $q \in \mathcal{P}(B)$ by 
\[
S_{\alpha, \mathcal{W}}(q) = \sup_{w \in \mathcal{W}} D_{\alpha}(w \Vert q).
\]
The order-$\alpha$ R\'enyi radius of $\mathcal{W}$ is defined as 
\[
S_{\alpha}(\mathcal{W}) = \inf_{q \in \mathcal{P}(B)} S_{\alpha, \mathcal{W}}(q).
\]
\end{Definition}

Given a joint distribution $\prob_{X,Y}$ of $(X,Y)$, one can consider the quantities 
$I_{\alpha}^{X\leadsto Y}, 
I_{\alpha}^{Y\leadsto X}, 
K_{\alpha}^{X\leadsto Y}, 
K_{\alpha}^{Y\leadsto X}$
and $J_{\alpha}(X;Y)$,  
as functions of an ``input'' distribution $\prob_X$ and the channel from  $X$ to $Y$ 
(i.e.,  the probability kernel formed by the conditional distributions of $Y$ given $X=x$). For example, one can consider 
$I_{\alpha}^{Y\leadsto X} = \inf_{\nu \in \mathcal{P}(S)} D_{\alpha} ( \prob_{X,Y} \Vert \nu \otimes \prob_{Y} )$ 
as a function of $\prob_{X}$ and the probability kernel formed by the conditional distributions of $Y$ given $X=x$. 
Under this interpretation, we first define various capacities. 

\begin{Definition} Given a channel $W$ from $X$ to $Y$, we define capacities of order $\alpha$
by
\begin{enumerate}
\item $\mathcal{C}^{X\leadsto Y}_{K, \alpha}(W) = \sup_{P \in \mathcal{P}(S)} K_{\alpha}^{X\leadsto Y}(P, W)$,
\item $\mathcal{C}_{J, \alpha}(W) = \sup_{P \in \mathcal{P}(S)} J_{\alpha}(P, W)$, and 
\item $\mathcal{C}^{X\leadsto Y}_{I,\alpha}(W) = \sup_{P \in \mathcal{P}(S)} I_{\alpha}^{X\leadsto Y}(P, W)$.
\end{enumerate}
\end{Definition} 

Theorem \ref{Theorem:sandwich} allows us to extend \cite[Theorem 5]{Ver15:ita} to include the capacity based
on the Lapidoth-Pfister mutual information.

%if $\alpha \geq 1$, and we fix a channel $W$ from $X$ to $Y$,
%\[
%\mathcal{C}^{X\leadsto Y}_{I,\alpha}(W)= \mathcal{C}_{J,\alpha}(W) =\mathcal{C}^{X\leadsto Y}_{K,\alpha}(W).
%\]
%Finally, using the symmetry of the Lapidoth-Pfister mutual information and Theorem \ref{Theorem:sandwich}, we deduce:

\begin{Theorem}\footnote{This theorem corrects Theorem V.1 in \cite{AM19:isit}.} \label{Theorem:cap}
Let $\alpha \geq 1$, and fix a channel $W$ from $X$ to $Y$. Then,
\[\begin{split}
\mathcal{C}^{Y\leadsto X}_{K,\alpha}(W) \leq \mathcal{C}^{X\leadsto Y}_{I,\alpha}(W)&= \mathcal{C}_{J,\alpha}(W) =\mathcal{C}^{X\leadsto Y}_{K,\alpha}(W)
 =S_{\alpha} (W) \leq \mathcal{C}^{Y\leadsto X}_{I,\alpha}(W).
\end{split}\]
\end{Theorem}
\begin{proof}
Theorem \ref{Theorem:sandwich} implies that,
\[
\sup_{P \in \mathcal{P(S)}} K_{\alpha}^{X\leadsto Y} (P,W) \leq \sup_{P \in \mathcal{P(S)}} J_{\alpha}  (P,W) \leq \sup_{P \in \mathcal{P(S)}} I_{\alpha}^{X\leadsto Y} (P,W),
\]
that is,
\[
\mathcal{C}^{X\leadsto Y}_{K,\alpha}(W) \leq \mathcal{C}_{J,\alpha}(W) \leq \mathcal{C}^{X\leadsto Y}_{I,\alpha}(W).
\]
It was shown by Csisz\'ar \cite{Csi95} in the finite alphabet setting (in fact, he showed this for all $\alpha>0$) 
that $\mathcal{C}^{X\leadsto Y}_{I,\alpha}(W)=\mathcal{C}^{X\leadsto Y}_{K,\alpha}(W)$. 
Nakibo\u{g}lu demonstrates $\mathcal{C}^{X\leadsto Y}_{I,\alpha}(W)=S_{\alpha} (W)$ in 
\cite{Nak19:1} and $\mathcal{C}^{X\leadsto Y}_{K,\alpha}(W)=S_{\alpha} (W)$ in \cite{Nak19:2} for abstract alphabets. 
Putting all this together, we have
\[
\mathcal{C}^{X\leadsto Y}_{I,\alpha}(W) = \mathcal{C}_{J,\alpha}(W) =\mathcal{C}^{X\leadsto Y}_{K,\alpha}(W) =S_{\alpha} (W).
\]
Finally, using the  symmetry of $J_{\alpha}$ and Theorem \ref{Theorem:sandwich} in a similar fashion again, we get
\[
\mathcal{C}^{Y\leadsto X}_{K,\alpha}(W) \leq  \mathcal{C}_{J,\alpha}(W) \leq \mathcal{C}^{Y\leadsto X}_{I,\alpha}(W).
\]
This completes the proof.
\end{proof}
%As mentioned before, the fact that $\mathcal{C}^{X\leadsto Y}_{I,\alpha}(W)=\mathcal{C}^{X\leadsto Y}_{K,\alpha}(W)$ was shown by Csisz\'ar \cite{Csi95}
%in the finite alphabet setting (in fact, he showed this for all $\alpha>0$) very recently, by Nakibo\u{g}lu \cite{Nak19:1}
%that both these quantities coincide with the R\'enyi radius $S_{\alpha} (W)$. What is new in Theorem~\ref{Theorem:cap} is that
%these quantities are shown to coincide with the capacity based on the
%Lapidoth-Pfister mutual information, and that they are bounded by capacities coming from Csisz\'ar and Sibson mutual
%informations with the roles of $X$ and $Y$ in the definition reversed. 
The two inequalities in the last theorem cannot be improved to be equalities, this follows from a counter-example communicated to the authors by C. Pfister. 

Since $J_{\alpha}(X ; Y)$ is no longer sandwiched between $K^{X \leadsto Y}_{\alpha}$ and $I^{X \leadsto Y}_{\alpha}$ when $\alpha \in (0,1)$ 
the same argument cannot be used to deduce the equality of various capacities in this case. However, when $\alpha \in [1/2,1)$, 
a direct demonstration proves that the Lapidoth-Pfister capacity of a channel equals R\'enyi radius when the state spaces are finite. 

\begin{Theorem}\label{thm: lpcapisrad}
Let $\alpha \in [\frac{1}{2},1)$, and fix a channel $W$ from $X$ to $Y$ where $X$ and $Y$ take values in finite sets $S$ and $T$ respectively. Then, 
\[
\sup_{P \in \mathcal{P}(S)} J_{\alpha}(P,W) = S_{\alpha}(\mathcal{W}).
\]
\end{Theorem}
\begin{proof}
We continue using integral notation instead of summation. Note that, 
\[
J_{\alpha}(X;Y) = - \max_{\mu \in \mathcal{P}(T)} -\beta \log \int_S e^{ \frac{1}{\beta} D_{\alpha}(W(x) \Vert \mu)} \d P (x) = \min_{\mu \in \mathcal{P}(T)} \beta \log \int_S e^{ \frac{1}{\beta} D_{\alpha}(W(x) \Vert \mu)} \d P (x),
\]
where $\beta = \frac{\alpha}{\alpha - 1}.$
We consider the function $f(P, \mu) = \beta \log \int_S e^{ \frac{1}{\beta} D_{\alpha}(W(x) \Vert \mu)} \d P (x)$ defined on $\mathcal{P}(S) \times \mathcal{P}(T).$ Observe that the function $g(P, \mu)= -e^{\frac{1}{\beta}f(P, \mu)} = - \int_S e^{ \frac{1}{\beta} D_{\alpha}(W(x) \Vert \mu)} \d P (x)$ has the same minimax properties as $f$. We make the following observations about this function. 
\begin{itemize}
\item \textit{$g$ is linear in $P$.}
%\\
%The term inside the logarithm $\int_S e^{ \frac{1}{\beta} D_{\alpha}(W(x) \Vert \mu)} \d P (x)$ is linear in $P$, therefore when composed with the quasi-convex function $x \mapsto \log x$, it remains quasi-convex. That is, $\log \int e^{ \frac{1}{\beta} D_{\alpha}(W(x) \Vert \mu)} \d P (x)$ is quasi-convex. Since $\alpha \in [1/2,1)$, $\beta<0$, and consequently $f$ is quasi-concave. 
\item \textit{$g$ is convex in $\mu$}.
\\
Follows from the proof in \cite[Lemma 17]{LP19}.
\item \textit{$g$ is continuous in each of the variables $P$ and $\mu$.}
Continuity in $\mu$ follows from continuity of $D_{\alpha}$ in the second coordinate (see, for example, in \cite{EH14}) whereas continuity in $P$ is a consequence of linearity of the integral (summation).
\end{itemize}
The above observations ensure that we can apply von Neumann's convex minimax theorem to $g$, and therefore to $f$ to conclude that
\[
\sup_{P} J_{\alpha}(P,W) = \sup_{P} \min_{\mu} f(P, \mu) = \min_{\mu} \sup_{P} f(P,\mu).
\]
%, quasiconcavity in $P$ follows from quasi-convexity of the logarith. Moreover, (from \cite{LP16:icsee}) we obtain convexity in $\mu$. Note that, $J_{\alpha}(X;Y) = \min_{\mu \in \mathcal{P}(T)} f(P, \mu)$, and $\mathcal{C}_{J,\alpha}(W) = \sup_{P \in \mathcal{P}(S)} \min_{\mu \in \mathcal{P}(T)} f(P, \mu) = \min_{\mu} \sup_{P}f(P, \mu)$. 
For a fixed $\mu$ however, $\sup_{P} f(P , \mu) = \sup_{P} \beta \log \int e^{\frac{1}{\beta} D_{\alpha}(W(x) \Vert \mu)} \d P(x) = \sup_{x} D_{\alpha} (W(x) \Vert \mu)$ (the RHS is clearly bigger than the LHS, for the other direction use measures $P = \delta_{x_{n}}$ where $x_{n}$ is a supremum achieving sequence for the RHS). This shows that when $\alpha \geq 1/2$ the capacity coming from the $J_{\alpha}(X;Y)$ equals the R\'enyi radius if the state spaces are finite.
\ 
\end{proof}

Though we do not treat capacities coming from Arimoto's mutual information in this paper due to its dependence on a reference measure, a remark can be 
made in this regard following B.~Nakibo\u{g}lu's \cite{Nak20:pers} observation that Arimoto's mutual information w.r.t. $\gamma$ of a joint distribution $(X,Y)$ 
can be written as a Sibson mutual information of some input probability measure $P$ and the channel $W$ from $X$ to $Y$ corresponding to $\prob_{X,Y}$. 
Let $X, Y$ denote the marginals of $(X,Y)$. As before there are reference measures $\gamma$ on the state space $S$ of $X$. Let $P$ denote the probability 
measure on $S$ with density $\frac{\d P}{\d \gamma} = e^{(1- \alpha) D_{\alpha}(\prob_{X} \Vert \gamma)} \left( \frac{\d \prob_{X}}{\d \gamma} \right)^{\alpha}$. Then a calculation shows that 
\[
I^{(\gamma)}_{\alpha} (X\leadsto Y) =  I_{\alpha}^{X\leadsto Y}(P, W).
\]
Therefore, it follows that if a reference measure $\gamma$ is fixed, then the capacity of order $\alpha$ of a channel $W$ calculated from Arimoto's mutual information will be less than the capacity based on Sibson mutual information (which equals the R\'enyi radius of $W$).

\subsection*{Acknowledgements} 
We thank C. Pfister for many useful comments, especially for pointing out an error in \cite{AM19:isit} that is corrected here,
and for a simpler proof of Theorem~\ref{thm: lpcapisrad}. 
We also thank B. Nakibo\u{g}lu for detailed suggestions, which include but are not limited to pointing us towards useful references, 
an easier proof of Theorem~\ref{Theorem:conditionalanddivergence}, and the observation discussed at the very end of the last section. 
 This research was supported in part by NSF grant DMS-1409504.

%\bibliographystyle{plain}
%\bibliography{pustak}

%%%%%%%%%%%%%%%%%%%%%%%%%%%%%%%%%%%%%%%%%%
\end{document}